\documentclass[runningheads]{article}
\usepackage{graphicx}
\pdfoutput=1
\usepackage{lipsum}
\usepackage{amssymb}
\usepackage{amsmath}
\usepackage{upgreek}
\usepackage{algorithm}
\usepackage{algorithmicx}
\usepackage{algpseudocode}
\usepackage{macros}
\usepackage{ioa_code_small-1}
\algrenewcommand\textproc{}
\usepackage{hyperref}
\algnewcommand\algorithmicforeach{\textbf{for each}}
\algdef{S}[FOR]{ForEach}[1]{\algorithmicforeach\ #1\ \algorithmicdo}
\usepackage{xspace}
\usepackage[utf8]{inputenc}
\usepackage{ textcomp }
\usepackage{tikz}
\usetikzlibrary{shapes,arrows,positioning,automata,calc,patterns}
\tikzset{node/.style={circle,draw,inner sep=0pt, minimum size=2em}}
\pgfdeclarelayer{background}
\pgfdeclarelayer{foreground}
\pgfsetlayers{background,main,foreground}
\usetikzlibrary{external}
\usepackage{subfigure}
\hypersetup{
  colorlinks=true,
  linkcolor=black,
  urlcolor=blue!70!black,
  citecolor=black,
}

\usepackage{ dsfont }
\usepackage{tikz}
\usetikzlibrary{shapes,arrows,positioning,automata,calc}
\tikzset{node/.style={circle,draw,inner sep=0pt, minimum size=2em}}
\pgfdeclarelayer{background}
\pgfdeclarelayer{foreground}
\pgfsetlayers{background,main,foreground}
\usetikzlibrary{external}

\usepackage{xcolor}
\definecolor{olive}{rgb}{0.3, 0.4, .1}
\definecolor{pinegreen}{cmyk}{0.92,0,0.59,0.25}

\usepackage{setspace}
\usepackage{color}
\usepackage{xcolor}
\usepackage[textwidth=0.15\textwidth,textsize=scriptsize]{todonotes}

 \usepackage{amsthm}

\newtheorem{theorem}{Theorem}[section]
\newtheorem{corollary}{Corollary}[section]
\newtheorem{lemma}{Lemma}[section]
\usepackage{authblk}

\author[1]{Alejandro Ranchal-Pedrosa}
\author[1,2]{Vincent Gramoli}
\affil[1]{University of Sydney, Sydney, Australia}
\affil[2]{Data61-CSIRO, Sydney, Australia}

\begin{document}
\title{Platypus: a Partially Synchronous Offchain Protocol for Blockchains}

\newcommand{\boxedtext}[1]{\fbox{\scriptsize\bfseries\textsf{#1}}}

\newcommand{\greenremark}[2]{
   \textcolor{pinegreen}{\boxedtext{#1}
      {\small$\blacktriangleright$\emph{\textsl{#2}}$\blacktriangleleft$}
}}
\newcommand{\myremark}[2]{
   \textcolor{blue}{\boxedtext{#1}
      {\small$\blacktriangleright$\emph{\textsl{#2}}$\blacktriangleleft$}
}}
\newcommand{\NewRemark}[2]{
   \textcolor{blue}{
      {\small$\blacktriangleright$\emph{\textsl{#2}}$\blacktriangleleft$}
}}

\newcommand\ARP[1]{\myremark{ARP}{#1}}
\newcommand\VG[1]{\myremark{VG}{#1}}
\newcommand\NEW[1]{\NewRemark{NEW}{#1}}
\newcommand\UPDATE[1]{\greenremark{UPDATE?}{#1}}

\maketitle

\begin{abstract}
Offchain protocols aim at bypassing the scalability and privacy limitations of classic blockchains
by allowing a subset of participants to execute multiple transactions outside the blockchain. 
While existing solutions like payment networks and factories depend on a complex routing protocol, 
other solutions simply require participants to build a \emph{childchain}, 
a secondary blockchain where their transactions are privately executed.
Unfortunately, all childchain solutions assume either synchrony or a trusted execution environment.

In this paper, we present Platypus a childchain that requires neither synchrony nor a trusted execution environment.
Relieving the need for a trusted execution environment allows Platypus to ensure privacy without trusting a central 
authority, like Intel, that manufactures dedicated hardware chipset, like SGX. 
Relieving the need for synchrony means that no attacker can steal coins by leveraging clock drifts or message delays to lure timelocks.
In order to prove our algorithm correct, we formalize the chilchain problem as a Byzantine variant of the classic Atomic Commit problem, where closing a childchain is equivalent to committing the whole set of payments previously recorded on the childchain ``atomically'' on the main chain. Platypus is resilience optimal and we explain how to generalize it to 
crosschain payments. 
\end{abstract}

\section{Introduction}
\label{sec:introduction}
One of the most important challenges of blockchains is scalability.
In fact, most blockchains consume more resources without offering better performance as the number of participants increases. 
Although some research results demonstrated that blockchain performance can scale with the number of participants~\cite{CNG18}, these rare solutions do not have other appealing properties, like privacy, built in.
As a result, blockchain extensions that offer scalability and privacy have been put forward,
 in what is known as \emph{Offchain protocols}. Examples of these protocols are state and payment \textit{channels}~\cite{poon2016bitcoin}
in which two parties can perform several offchain payments with one another; \textit{channel networks}~\cite{poon2016bitcoin} 
that allow users to relay payments in a network of channels; \textit{channel factories}~\cite{2019scalable}
that open multiple channels in one transaction, saving storage and fees; and \textit{childchains}~\cite{poon2017plasma}
which are secondary blockchains pegged to the existing, so called ``parent'', blockchain.

By relying on offchain computation, all these protocols avoid communicating and/or storing some information directly in the blockchain---hence bypassing the performance bottleneck of the blockchain but also limiting transparency of selected transactions to ensure privacy.
Whereas channels, channel networks and channel factories offer private and fast payments, they can only perform payments if users have an existing route of channels with one another. As a result their scalability and privacy are actually subject to proper handling of the network topology and vulnerable to routing attacks~\cite{2019difficulty}.

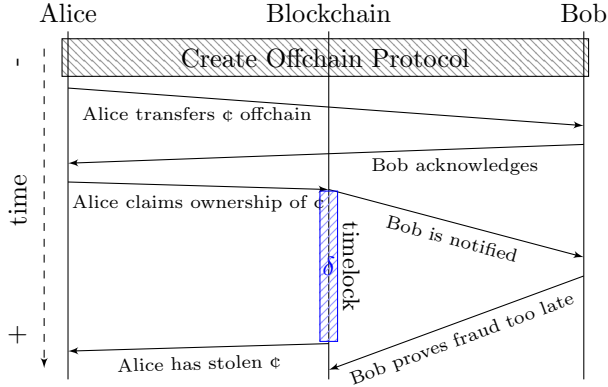
\begin{figure}[t]
  \centering
 \begin{tikzpicture}[node distance=2cm,auto,>=stealth']
    \node[] (bob) {Bob};
    \node[left = of bob] (bck) {Blockchain};
    \node[left = of bck] (alice) {Alice};
    \node[left = of bck,xshift=-2em,yshift=-1em] (time1) {};
    \node[below of=bck,yshift=4em,xshift=-0.13em,rectangle,minimum width=7cm,minimum height=0.5cm,pattern=north west lines, pattern color=black!40, draw] (creation) {Create Offchain Protocol};
    \node[below of=bob, node distance=5cm] (bob_ground) {};
    \node[below of=alice, node distance=5cm] (alice_ground) {};
    \node[below of=bck, node distance=5cm] (bck_ground) {};
    \node[below of=time1,node distance=4.5cm] (time2) {};
    %
    \draw (alice) -- (alice_ground);
    \draw[dashed,-latex'] (time1) -- node [midway,rotate=90,yshift=1em,xshift=-1em]{time} node [near start,rotate=90,yshift=1em,xshift=2em]{-} node [near end,rotate=90,yshift=1em,xshift=-2.5em]{+} (time2);    
    \draw (bob) -- (bob_ground);
    \draw (bck) -- (bck_ground);
    \draw[-latex'] ($(alice)!0.20!(alice_ground)$) -- node[below,near start]{\scriptsize Alice transfers \textcent~offchain} ($(bob)!0.30!(bob_ground)$);
    \draw[-latex'] ($(bob)!0.35!(bob_ground)$) -- node[below,midway, near start]{\scriptsize Bob acknowledges} ($(alice)!0.40!(alice_ground)$);
    \draw[-latex'] ($(alice)!0.45!(alice_ground)$) -- node[below,midway]{\scriptsize Alice claims ownership of \textcent} ($(bck)!0.47!(bck_ground)$);
    \draw[-latex'] ($(bck)!0.47!(bck_ground)$) -- node[below,sloped, midway]{\scriptsize Bob is notified} ($(bob)!0.65!(bob_ground)$);
    \draw[-latex'] ($(bob)!0.7!(bob_ground)$) -- node[below,sloped,midway,align=center]{\scriptsize Bob proves fraud too late} ($(bck)!0.95!(bck_ground)$);
    \draw[-latex'] ($(bck)!0.88!(bck_ground)$) -- node[below,sloped,midway]{\scriptsize Alice has stolen \textcent} ($(alice)!0.9!(alice_ground)$);
    \node[below of=bck,rectangle,yshift=-3.9em,minimum height=2cm,pattern=north east lines, pattern color=blue!40, draw=blue] (period) {};
    \node[right of=period,xshift=-2cm] (dt) {\color{blue}$\delta$};
    \node[right of=period,xshift=-1.7cm,rotate=-90] (dt) {\small timelock};
\end{tikzpicture}
\caption{Alice can steal Bob's coin if Bob messages are delayed such that Bob's reply takes longer than the timelock $\delta$. 
}
\label{fig:01}
\vspace{-2em}
\end{figure}

There is, therefore, great interest in designing proper childchain protocols that allow blockchains to host the creation and destruction of other smaller blockchains that depend on 
them. 
Unfortunately, as far as we know all childchains~\cite{back2014enabling,gazi2019proof} use timelocks that only work under the assumption that the communication is \emph{synchronous}, in that every message gets delivered in less than a known bounded amount of time~\cite{dwork1988consensus}.
This assumption is easily violated in large networks like the internet, due to either natural disasters or human misconfiguration of BGP tables for example. 
But more dramatically, assuming synchrony exposes childchains to various attacks, like Denial-of-Service or Man-in-the-Middle, that are common practice to double spend~\cite{EGJ18}.

To illustrate the problem of chilchains,
consider an execution using timelocks illustrated in Figure~\ref{fig:01} in which time increases from top to bottom. 
First, Alice transfers \textcent{} coins to Bob offchain before Bob acknowledges the transfer. 
Bob can then take actions in response to this transfer thinking, wrongly, that he will have sufficient time to prove the fraud if Alice tries to claim back the coins.
Let us consider that Alice is Byzantine (or malicious) and claims back the ownership of the coins, which triggers a timelock, a safe guard delay during which the coins are locked to give an opportunity to other participants to prove fraudulent activity before the coins are transferred back. 
As part of the protocol, Bob gets notified but due to an unforeseen delay, does not manage to prove the fraud before the end of the timelock. Its \textcent{} coins are thus stolen. 


In this paper, we propose \textit{Platypus}, the first offchain protocol for childchains that does not assume synchrony.
To this end, we formalize offchain protocols as a Byzantine fault tolerant atomic commit of multiple transfers. 
Platypus exploits \emph{partial synchrony}~\cite{dwork1988consensus}---also called eventual synchrony---where messages take  unknown amounts of time to be delivered.
We prove the correctness and resilience optimality of Platypus and discuss applications to crosschain payments.
Finally, we show that the time, message and communication complexities of Platypus is lower or comparable to consensus algorithms.

The rest of this document is structured as follows: Section~\ref{sec:bac} provides some background and preliminary definitions, Section~\ref{sec:mod} introduces our model and Section~\ref{sec:pro} presents the Platypus protocol. We prove Platypus correct in Section~\ref{sec:cor}. We analyse the complexity of Platypus in Section~\ref{sec:analysis}. Section~\ref{sec:dis} discuss applications of Platypus to crosschain payments. Section~\ref{sec:rel}  presents the related work and we conclude in Section~\ref{sec:conc}.

\section{Preliminaries}
\label{sec:bac}
In this section we discuss the background and introduce important notations.
\begin{itemize}
  
  \item \textit{Consensus protocol} We require a relaxed validity property for consensus, in which Byzantine processes proposing a valid value can be taken. consensus protocol between $P$ processes is any protocol that satisfies the following properties:
\begin{itemize}
\item Termination. Every correct process $p$ eventually decides on a value.
\item Validity. The value decided by a correct process verifies a predefined predicate valid.
\item Agreement. No two different processes decide on a different value.\\
\end{itemize}
\item \textit{Blockchain.} Inspired by~\cite{gazi2019proof}, we refer to a blockchain $\Omega= \langle b_i \rangle$ as a distributed ledger that builds upon a consensus protocol in order to add \textit{blocks} $b_i$. We denote $\Omega[i]$ as the $i^{th}$ block of $\Omega$, with $\Omega[-i]$ being the $i^{th}$ latest block of $\Omega$. Transactions are added to a block that is then written in $\Omega$. The set of processes $V$ that agree on the transactions to be written in the next block are the \textit{validators}. As we assume the presence of a consensus protocol, we consider that the blockchain cannot fork due to a disagreement, we call it \emph{unforkable}.
\item \textit{Transactions.} Similar to the model of~\cite{gurcan2018cancellation}, a transaction is a tuple $tx=\langle I, O\rangle$ where $I$ is a list of inputs and $O$ a list of outputs. Outputs are stored in an Unspent Transaction Output (UTXO) pool until a transaction that consumes it as one of its inputs gets written in $\Omega$. We model the outputs as $o_i=\langle s_i,$\textcent$_{o_i}\rangle $ where the set $s_i=\{(p_i,conds_{p_i})\}$ defines the conditions $conds_{p_i}$ for the process $p_i$ to spend the coin \textcent$_{o_i}$ (\textcent$_{o_i} \geq 0$). 
%
In order to spend a coin, the associated conditions $conds_{p_i}$ must be fulfilled
so that only one process, among multiple candidate ones, can spend this coin.
\item \textit{Ownership.} We say that process $p_i$ owns coin \textcent$_i$ if there exists a list of conditions $conds_{p_i}$ such that $p_i$ can spend \textcent$_i$. 
As such, let $\mathds{C}$ be the set of coins, $\mathcal{T}$ the set of discrete timeslots (such as blockheight), and $P$ the set of processes, then ownership is a function $\varphi: \mathds{C}\times\mathcal{T}\rightarrow P$ that takes a coin and returns its owner 
at a particular time. 
  \item \textit{Transferring coins.} A transaction may transfer one or more coins. We refer to $p_i$ transferring a coin \textcent$_i$ to $p_j$ if a process $p_i$ spends it to $p_j$. We can define a transfer of a coin as a change of ownership. That is, let $a,\,b\in P$ and let \textcent$_i\in \mathds{C},\,t_j\in\mathcal{T}$, such that $\varphi(\text{\textcent}_i,t_j)=a$, then the transfer relation $\ms{TR}$ to b is such that {\small $a\, \ms{TR}_{t_{j+1},\text{\textcent}_i}\, b \iff \varphi(\text{\textcent}_i,t_{j+1})=b$}. 
 Notice we can define the transitive closure of the transfer operation as follows:\\
  \begin{equation}
    \hspace{-0.5em}
    \small
    \ms{TR}^+=  \left \{\begin{aligned}(a,b)\in P^2\,:\; &\exists \text{\textcent}_i\; s.t. \; \varphi(\text{\textcent}_i,t_j)=a \textbf{ and }\varphi(\text{\textcent}_i,t_k)=b,\\& \textbf{ for some } t_j,t_k\in\mathcal{T},\; j<k. \end{aligned}\right \}
  \end{equation}

\end{itemize}


\section{Model}
\label{sec:mod}

We define in this section some assumptions and concepts for our model:
\begin{itemize}
\item \textit{Partial synchrony and failures.}  Our model relies on a partially synchronous network~\cite{dwork1988consensus}, i.e. a network in which each message has an unknown communication bound. In other words, there exists an unknown time where the network stabilizes and after which every messages is delivered in a bounded amount of time. We also assume that strictly less than $n/3$ processes participating in an offchain protocol are \emph{Byzantine} in that they may fail arbitrarily and that other processes are \emph{correct} (as we detail in the adversary model below).

\item \textit{Accounts.} We define an account $a$ as an instance of only one process $p_i$, $\rho(a)=p_i$, where $\rho(a)$ is a function that returns the process that controls account $a$. An account belongs to a particular blockchain, one account is controlled by only one process, but one process can have multiple accounts, either in the same or in different blockchains. 
  \item \textit{Threshold signatures.} Our model requires accounts to authenticate with a cryptographic primitive enabling non-interactive aggregation, such as those of~\cite{gazi2019proof,boneh2018compact}. For simplicity and without loss of generality we assume that accounts are not reusable. In particular, the same coins 
  should not go back to the same process in the same account, to prevent a variant of the ABA problem (see Section~\ref{sec:dis}). In the remainder, we abuse the term process as an account that the process owns, unless stated otherwise.

  \item \textit{Minimal transfers.} Given a sequence $seq=\{u_{j}\, \ms{TR}_{t_{j+1},\text{\textcent}_i}\, u_{j+1}\}_{j=c}^{d-1}$ of transfers over some timerange $[t_c,t_d]$,  between creation time $t_c$ and destruction time $t_d$, for coin $\text{\textcent}_i$, we refer to the minimal transfer as the single transfer $u_c\,\ms{TR}_{\text{\textcent}_i} u_d$, which is always an element of the transitive closure. For a set of operations defined over all coins within a timerange $[t_c,t_d]$, we denote the minimal transfer set $\ms{TR}^-$ as the set of all minimal transfers, which is at least a set of idempotent transfers of the form $a\, \ms{TR}\, a$.  
  \item \textit{Offchain problem.} 
%
%
Given a blockchain $\Omega$ of $P$ processes, 
the offchain protocol consists of executing a sequence $seq$  of transfers ``off chain''.
First processes $Q \varsubsetneq P$ must create an offchain protocol $\Gamma$ by writing a transaction in the original chain $\Omega$---effectively depositing funds from $\Omega$ into $\Gamma$.
Then they transfer coins ``off chain'' among themselves using $\Gamma$. 
Finally they can destroy this protocol $\Gamma$.
To this end, the offchain protocol consists of at least two main procedures, \emph{creation} and \emph{bulk close}. (We will explain later how the participation in $\Gamma$ is made dynamic using splice in and splice out procedures to accept new participants and for existing participants to leave $\Gamma$, respectively.)
After a series of transfers in $\Gamma$,  
processes can propose to bulk close it by \emph{proposing COMMIT}.
Processes \emph{decide to COMMIT} in that they effectively agree to accept these transfers and to close and destroy $\Gamma$ or \emph{decide to ABORT} in that they 
disagree with the transfers and refuse to close $\Gamma$.
Hence, a protocol solving the \emph{Offchain problem} must satisfy the following properties:
\begin{itemize}
    \item Termination: every correct process decides COMMIT or ABORT on some sequence of transfers $seq$ for which some process proposed COMMIT.
    \item Agreement: no correct process decides COMMIT on  two different sequences $seq$ and $seq'$.
    \item ABORT-Validity: if a correct process proposes ABORT for a sequence $seq$, then all correct processes decide ABORT for this sequence $seq$.
    \item COMMIT-Validity: if no correct process proposes ABORT for sequence $seq$ for which some process proposed COMMIT, then all correct processes decide COMMIT for sequence $seq$.
 \end{itemize}


Notice that, in our definition, aborting is implicit and proposing ABORT is not an input of our algorithm as we will see in Algorithm~\ref{alg:val}.
In particular, COMMIT-Validity can be ensured by requiring a process to provide a valid Proof-of-Fraud (PoF) when proposing to abort, the invalidity  of the PoF allows correct processes to ignore the ABORT proposal and its validity guarantees that all correct will observe this PoF. 
%
Finally note that our termination does not imply that the offchain protocol gets closed. Instead, it means that all correct processes decide either COMMIT and closes  the protocol or ABORT and not closing the protocol. 
This is not a problem since, as we will explain in Algorithm~\ref{alg:spou}, any correct process can cash out the coins that it knows it owns at any moment.

    
    


    

In order to achieve privacy, we need another property stating that some decisions of the offchain protocol do not have to be written in the blockchain:
  \begin{itemize}
  \item COMMIT-Privacy/Lightness: If correct processes decide COMMIT on a sequence $seq$ 
    of transfer operations made in $\Gamma$ between $t_c$ and $t_d$, then $\forall p \in P\backslash Q,\, p$ 
    only learns/stores  $\ms{TR}^-$, the minimal transfer set of $seq$.
  \end{itemize}
\item \textit{Childchain.} A \emph{childchain} $\Psi$ is a particular class of offchain protocol
in that it is a blockchain $\Psi$ that is created by another blockchain $\Omega$, known as its \textit{parentchain}, and that implements an Offchain protocol $\Gamma$.

\item \emph{Adversary model.}
\label{sec:adv}
  We consider an adversary $F$ such that:
  \begin{itemize}
    \item $F$ can control the network to read or delay messages, but not to drop them.  
    \item $F$ can take full control and corrupt a coalition of $f$ processes, learning its entire state (stored messages, signatures, etc.). It takes control of receiving and sending all their messages. This adversary can also guess in advance the estimate value of any correct process in any round. Furthermore, it delivers the messages from correct nodes instantly, and its messages are delivered instantly by any correct nodes.
    \item $F$ cannot forge signatures of processes outside the coalition $f$.
    \item  We define $t_0$
    and $t_1$
    two thresholds for Byzantine behavior. That is, the coalition must be such that $f \leq t_0$ and $f \leq t_1$.
    \end{itemize}
    \end{itemize}

\section{Secure Childchains Without Synchrony}
\label{sec:pro}

In this section we present Platypus, a novel childchain protocol that solves the offchain problem without assuming synchrony. 
Platypus consists of both an offchain protocol and a childchain that are denoted respectively $\Gamma$ and $\Psi$ in the remainder of the paper. 
Given a parentchain $\Omega$, processes can use the protocol $\Gamma$ by depositing funds from $\Omega$ to $\Psi$, that effectively creates the childchain. Then transfers can be done directly on $\Psi$  ``off chain'' before the bulk close happen.  

The parentchain $\Omega$  and childchain $\Psi$ have a set $P_\Omega$ of $|P_{\Omega}|=n_p$ users and $P_\Psi$ of $|P_\Psi|=m_p$ users, respectively, with a set $V_\Omega\subseteq P_\Omega$ of $|V_\Omega|=n_v$ validators and a set  $V_\Psi\subseteq P_\Psi\subseteq P_\Omega$ of $|V_\Psi|=m_v$ validators, respectively. 
Note that $m_v$ is the number of all processes joining the Platypus protocol.
As mentioned before, however, we assume that at most  $t_1 = \lceil m_v/3 \rceil -1$ among them are Byzantine.
%
Although we do not provide an implementation of the blockchain $\Psi$ we assume that $\Psi$ is secure (i.e. it uses deterministic consensus to not fork): a blockchain assuming partial synchrony and $\lceil m_v/3 \rceil -1$ Byzantine processes among $m_v$ processes like Red Belly~\cite{CNG18,crain2018dbft} can be used here.


\subsection{Overview}
 $\Gamma$ is depicted in three main procedures: a creation (Alg.~\ref{alg:plcr}), a bulk close  (Alg.~\ref{alg:plcl}) and an abort (Alg.~\ref{alg:val}). (Splice in and splice out procedures are deferred to Section~\ref{sec:dis}).
Processes can ABORT or COMMIT sequences of transfers done in $\Psi$. In particular, a process proposes ABORT by creating an abort transaction  (and sharing it) in line~\ref{line:vPoF} of Alg.~\ref{alg:val} and proposes a COMMIT at line~\ref{line:cor} of Alg.~\ref{alg:plcl}.
A process decides COMMIT 
at line~\ref{line:par3} (Alg.~\ref{alg:plcl}) only after $m_0$ processes propose COMMIT  and decides ABORT at line~\ref{line:par2} (Alg.~\ref{alg:plcl}) only when there exists a valid abort transaction.

\subsection{Creating a Platypus Chain}
  Users can create a Platypus chain by publishing a transaction on $\Omega$. After that transaction is finalized, the funds referred to in this transaction are locked and ready to be used by the $\Gamma$. 
\label{subsec:plcr}
In general, a Platypus creation transaction ($tx_{plcr}$) is a transaction that:
\begin{itemize}
  \item Has a new Platypus id ($\ms{plid}$) that uniquely identifies it.
  \item Specifies a consensus protocol for the Platypus Blockchain to decide on a new block. W.l.o.g., we assume DBFT~\cite{crain2018dbft} to be the default protocol.
  \item Specifies a number $m_0>f$ of validators required to create $\Psi$. For simplicity and to match with the optimal result (see Theorem~\ref{the:imp}), we choose $m_0=\lfloor 2m_v/3\rfloor +1$. 
    
  \item Defines a new function $\lit{abort(...)}$ that specifies when a user can decide ABORT on the protocol (such as a Platypus bulk close transaction being aborted).
    
  \item Specifies a set of processes and their balances that go in the Platypus Blockchain through this transaction.
    
  \item Once written in $\Omega$, the funds can only be spent in $\Psi$.
  \end{itemize}
  Algorithm~\ref{alg:plcr} shows the protocol to create a Platypus chain. The call to $\lit{num\_signers}(tx)$ returns the amount of signers of $tx$, while the call to $\lit{verify(tx,\{msg\})}$ verifies the validity of the transaction and signed messages. 
  We define two main interactions of the Platypus protocol with both the childchain and the parentchain: sending transactions and reading transactions. The Platypus protocol $\Gamma$ sends transactions to $\Omega$ or $\Psi$ by invoking $\lit{send(\{\Omega,\Psi\},tx)}$ and $\lit{acsend(\{\Omega,\Psi\},tx)}$. In the former, the function returns once the transaction is written in the corresponding blockchain, or a transaction that spent the same funds has been written (meaning this transaction became invalid), while the latter returns ABORT or COMMIT and the respectively written transaction in a response message. This response is received by all validators as it is a result of the Platypus Blockchain. Reading transactions is performed by the call to $\lit{is\_written(\{\Omega,\Psi\},tx)}$ that returns True or False depending on if the transaction was written or not in the Blockchain.
Each of the messages are signed, to prevent Byzantine nodes from adding third parties without their agreement. 



  \begin{algorithm}[H]
    \caption{Platypus creation procedure}
    \label{alg:plcr}
    \begin{algorithmic}[1]
      \small 
    \Statex $\vartriangleright$ State of the algorithm
    \Statex $\Omega$, the parentchain
    \Statex $\Gamma$, the Platypus protocol
    \Statex $P_{\Omega}$, the set of processes in the parentchain
    \Statex $P_{\Psi}\gets \bot$, the set of processes in the Platypus chain
    \Statex $V_{\Psi}\gets \bot$, the set of validators in the Platypus chain
    \Statex $m_v$, the amount of validators required in $\Psi$
    \Statex $\mathds{C}_i$, coins that belong to process $p_i$
    \Statex $job_i$, boolean defining if $p_i$ is $\lit{VALIDATOR}$ or just $\lit{USER}$
    \Statex $\ms{plid}$, the Platypus chain identifier
    \Statex $\ms{msg}_{i}=\langle \mathds{C}_i,\ms{plid},job_i\,\sigma_i\rangle$, signed message to join.
    \Statex $\sigma_i$, signature of $msg_i$ by $p_i$
    \Statex $tx_{plcr}\gets \bot$, the Platypus creation transaction
    \Statex \rule{0.4\textwidth}{0.4pt}
    \Statex $\vartriangleright$ PHASE 1: process $p_0$ initiates request
    \State $\ms{msg}_{0}\gets \lit{sign(\langle \mathds{C}_0,\ms{plid},job_0\rangle)}$\label{line:job1}
    \State $\lit{multicast}(\ms{msg}_{0})$ to $P_{\Omega}$\label{line:job2}
    \Statex \rule{0.4\textwidth}{0.4pt}
    \State  $\vartriangleright$ PHASE 2: Rest of processes who want to join reply
    \State \textbf{when} $\ms{msg}_0$ is received from $p_0$
    \State $\ms{msg}_{i}\gets \lit{sign(\langle \mathds{C}_i,\ms{plid},job_i\rangle)}$
    \State $\lit{multicast}(\ms{msg}_i)$ to $P_{\Omega}$
    \Statex \rule{0.4\textwidth}{0.4pt}
    \Statex $\vartriangleright$ PHASE 3: Validator $p_i\in V_{\Psi}$ gathers enough validators
    \State \textbf{when} $\ms{msg}_j$ is received from $p_j$ \textbf{and} $p_j\not\in P_\Psi$
    \State $\{P_\Psi,\mathds{C}_{P_\Psi}\}\gets \{P_\Psi\cup\{p_j\},\;\mathds{C}_{P_\Psi}\cup \ms{msg}_j.\mathds{C}_j\}$
    \SmallIf{$\ms{msg}_j.job_j=\lit{VALIDATOR} $ \textbf{and} $p_j\not\in V_\Psi$}{}\label{line:coi}
    \State $\{V_\Psi,\mathds{C}_{V_\Psi}\}\gets \{V_\Psi\cup\{p_j\},\;\mathds{C}_{V_\Psi}\cup \ms{msg}_j.\mathds{C}_j\}$
    \SmallIf{$|V_\Psi|= m_v$}{}\label{line:plcr1}\Comment{Enough validators to start transaction}
    \State $tx_{plcr} \; \gets \lit{createPlatypusTx(}\mathds{C}_{P_\Psi},\mathds{C}_{V_\Psi},\ms{plid}\lit{)}$
    \State $tx_{plcr} \; \gets \lit{sign_i(}tx_{plcr}\lit{)}$
    \State $\lit{multicast}( tx_{plcr},\{\ms{msg}_k\}_{p_k\in P_\Psi} )$ to $V_{\Psi}$ \label{line:comcop}
    \EndSmallIf
    \EndSmallIf
    \rule{0.4\textwidth}{0.4pt}
    \Statex $\vartriangleright$ PHASE 4: $p_i\in V_{\Psi}$ signs and broadcasts until it gets enough signatures
    \State \textbf{when} $(tx_{plcr},\{\ms{msg}_j\}_{p_j\in P_\Psi} )$  
    \label{line:plcrcond}
is received \textbf{and} 
    \textbf{not} $\lit{is\_written}(\Omega,tx_{plcr},\ms{plid})$ \Comment{if $tx_{plcr}$ with $\ms{plid}$ not written in $\Omega$}
    \SmallIf{$\lit{verify}(tx_{plcr},\{\ms{msg}_j\})$}{$tx_{plcr} \; \gets \lit{sign_i(}tx_{plcr}\lit{)}$} \label{line:plcrver} \label{line:plcrval}
    \SmallIf{$\lit{num\_signers(}tx_{plcr}\lit{)} < \lfloor 2m_v/3\rfloor +1$}{}\label{line:plcrsig}
    \State $\lit{multicast}(tx_{plcr},\{msg_j\})$ to $V_{\Psi}$\label{ref:plcrcont}
    \EndSmallIf
    \SmallElse{ $\Gamma.\lit{send(}\Omega,tx_{plcr}\lit{)}$} \Comment{enough signatures}\label{line:chi}
    \EndSmallElse
    \EndSmallIf
  \end{algorithmic}
  
\end{algorithm}
\subsection{Closing a Platypus Chain}
\label{subsec:plcl}

  A Platypus bulk close transaction splices all funds out of the Platypus Blockchain without compromising its security (agreement), and without requiring all validators to join together in its destruction (termination). It is still a normal transaction in the Platypus blockchain, meaning that it requires enough validators $m_0$ agreeing to writing it in $\Psi$. Algorithm~\ref{alg:plcl} shows the protocol to bulk close a Platypus chain. A Platypus bulk close transaction signed by some processes returns back the updated balances of all processes in the parentchain $\Omega$, unless it is aborted. Once written in both $\Psi$ and $\Omega$, the coins can be spent only on $\Omega$.

  \begin{algorithm}[H]
    \caption{Platypus bulk close procedure}
    \label{alg:plcl}
    \begin{algorithmic}[1]
      \small 
      \Statex  $\vartriangleright$ State of the algorithm
      \Statex $\Omega$, $\Psi$, $\Gamma$, the Blockchain, Platypus Blockchain and protocol
      \Statex $P_{\Psi}$,$V_{\Psi}$,  the set of processes and validators in $\Psi$
      \Statex $\mathds{C}_i$, the coins that belong to process $p_i$
      \Statex $tx_{plcl}\gets \bot$
      \Statex \rule{0.4\textwidth}{0.4pt}
      \Statex $\vartriangleright$ PHASE 1: Some process $p_0$ creates and broadcasts
      \State$tx_{plcl} \; \gets \lit{createBulkCloseTx(}\mathds{C}_{P_\Psi}\lit{)}$
      \State $tx_{plcl} \; \gets \lit{sign_i(}tx_{plcl}\lit{)}$
      \State $\lit{multicast}(\lit{tx_{plcl}})$ to $V_{\Psi}$
    
    \Statex \rule{0.4\textwidth}{0.4pt}
    \Statex $\vartriangleright$ PHASE 2: $p_i\in V_{\Psi}$ signs and broadcasts transaction\\
      \textbf{when} $tx_{plcl}$ 
      is received \textbf{and} 
      \textbf{not} $\lit{is\_written}(\Psi,tx_{plcl})$
      
      \State $\lit{verify(}tx_{plcl}\lit{)}$
      \State $tx_{plcl} \; \gets \lit{sign_i(}tx_{plcl}\lit{)}$
      \SmallIf{$\lit{num\_signers(}tx_{plcl}\lit{)} < \lfloor 2|V_\Psi|/3\rfloor +1$}\label{line:sig}{}
      \State $\lit{multicast}(tx_{plcl})$ to $V_{\Psi}$
           \EndSmallIf
            \SmallElse{} 
            {$r\gets\Gamma.\lit{acsend(}\Psi,tx_{plcl}\lit{)}$ \Comment{Get back $tx_{plcl}$, or $tx_{Abort}$} \label{line:cor}}
\EndSmallElse
\rule{0.4\textwidth}{0.4pt}
    \Statex $\vartriangleright$ PHASE 3: $\Gamma.\lit{acsend}(\Psi,tx_{plcl})$ generates a response, any $p_i$ can send to $\Omega$
    \Statex \textbf{when} $r$ is received
      \SmallIf{$r.type=\lit{ABORT}$}{$\Gamma.\lit{send(}\Omega,r.tx_{Abort}\lit{)}$\label{line:par3}}
      \EndSmallIf
      \SmallElseIf{$r.type=\lit{COMMIT}$}{$\Gamma.\lit{send(}\Omega,r.tx_{plcl}\lit{)}$\label{line:par2}}
      \EndSmallElseIf
  \end{algorithmic}
\end{algorithm}
\subsection{Aborting a Closing Attempt}
A Platypus bulk close transaction with insufficient signatures can either be a valid, ongoing Platypus bulk close, or an attempt to commit fraud. To prevent this, and guarantee termination and ABORT-validity, we introduce the abort transaction.

A transaction may be invalid if it spends a coin formerly owned by a user, but that was transferred to another user later on in $\Psi$.
The abort function runs for every Platypus bulk close transaction received that is not valid, i.e. that spends some input already spent. If the transaction is not valid due to signatures not matching, then it will not be written in the parentchain, so the abort function ignores this case.

Therefore, a user can see a transaction $tx$ is not valid if an old owner claims ownership of a spent coin in $tx$, as checked by $\lit{coins\_spent(...)}$, shown in Algorithm~\ref{alg:val}. Notice that, while a COMMIT requires $m_0$ validators to commit to the transaction (such as a Platypus bulk close transaction), any process $p\in P_\Psi$ can create a valid abort transaction. The call to $\lit{extract\_spent}(tx)$ returns the coins that were spent. The call to $\lit{get\_block\_min\_blockheight(}\mathds{C}_S\lit{)}$ returns the block of minimum blockheight out of all the blocks that store a transaction spending each of the spent coins, i.e. Proofs-of-Fraud (PoFs). Finally, $\lit{validators}(b/tx)$ returns the set of validators of block $b$ or transaction $tx$.

  \begin{algorithm}[H]
    \caption{Abort procedure}
    \label{alg:val}
    \begin{algorithmic}[1]
      \small 
      \Statex  $\vartriangleright$ State of the algorithm
      \Statex $\Omega$, $\Psi$, $\Gamma$, the Blockchain, Platypus Blockchain and protocol.
      \Statex $\ms{\mathds{C}_S}\gets \bot$, the subset of spent coins from $\mathds{C}$
      \Statex $\ms{b_p}\gets\bot$, integer s.t. $\Psi[\ms{b_p}]$ proves some coin was spent
      \Statex $\ms{vPoF}\gets \bot$, Proofs-of-Fraud of validators
      \Statex $tx_{abort}\gets\bot$
      \Statex      \rule{0.4\textwidth}{0.4pt}
    \Function{$\lit{abort}$($tx_{plcl}$)}{}
    \State $\ms{\mathds{C}_S}\gets \lit{extract\_spent(}tx_{plcl}\lit{)}$
      \State $\ms{b_p}\gets \lit{get\_block\_min\_blockheight(}\ms{\mathds{C}_S}\lit{)}$\label{line:blockproof}
      \State $\ms{vPoF} \gets \emptyset$
      \ForEach{$\ms{block}$ \textbf{in} $\Psi[b_p,...,-1]$}
      \State $\ms{vPoF}.\lit{append(}\lit{validators(}\ms{block}\lit{)}\cap \lit{validators(}\ms{tx_{plcl}}\lit{)}\lit{)}$
      \EndFor
      \State $tx_{abort} \; \gets \lit{createAbortTx(}\ms{tx_{plcl},b_p,vPoF}\lit{)}$\label{line:vPoF}
      \State $\Gamma.\lit{send(}\Omega,tx_{abort}$)
      \State $\Gamma.\lit{send(}\Psi,tx_{abort}$)
      \EndFunction
      \Statex \rule{0.4\textwidth}{0.4pt}
      \Statex $\vartriangleright$ all $p_i\in P_\Psi$ run abort when receiving any invalid $tx_{plcl}$
      \State \textbf{when} $tx_{plcl}$ is received
      \If{coins\_spent($tx_{plcl}$)} \Comment{some coins in $tx_{plcl}$ were spent, invalid}
      \State abort($tx_{plcl}$)
      \EndIf
    \end{algorithmic}
  \end{algorithm}

    Intuitively, this algorithm proves invalidity by iterating through $\Psi$, looking for validators that validated both this bulk close and some progress in $\Psi$ that conflicts with it (i.e. some blocks that spent some of the coins). This set of validators is the set of \textit{fraudsters}. Other validators that only validated the transaction might simply have had an old view of the Platypus chain, under the partially synchronous model. Nonetheless, the existence of such block is enough to create the abort transaction, even if the set of fraudsters is empty.
    
    The iteration starts from the block with minimum blockheight of all the Blocks that show that some coin \textcent$\,$ was transferred from $p_i$ to $p_j$, for some $p_i$ that claims ownership of \textcent$\,$  $tx_{plcl}$, in line \ref{line:blockproof}. The algorithm then continues to account for fraudsters. From that block, the process iterates forward in $\Psi$, gathering some possible validators that may have validated both $tx_{plcl}$ and conflicting blocks, i.e. looking for fraudsters.


\section{Correctness}
\label{sec:cor}
In this Section, we analyze the correctness of the protocol. To consider its correctness, we must prove that the protocol satisfies all the properties of Offchain protocols, as defined in Section~\ref{sec:mod}. We start by proving the proper bootstrapping of a Platypus chain, i.e. the adversary never locks the algorithm nor gains enough relative power in the validators set. Then, we prove the properties of Offchain protocols when closing a Platypus chain.

\begin{theorem}
  \label{the:plcrter}
  Algorithm~\ref{alg:plcr} terminates.
\end{theorem}
\begin{proof}
  The algorithm waits for enough Platypus Creation signed messages $\{msg_i\}$ from validators (line~\ref{line:plcr1}) and to get enough signatures from validators for the Platypus Creation transaction (line~\ref{line:plcrsig}). Since we assume there are at least $m_v$ processes that explicitly state that want to get in $\Psi$ as validators (Section~\ref{sec:adv}), the first condition is met to terminate. That is, a correct process will eventually produce and broadcast a valid Platypus Creation transaction with signed $\{msg_i\}$ messages of each of the users that committed to participate in such transaction.

  As for the signatures of the $tx_{plcr}$ transaction, notice only $m_0$ of the $m_v$ are required to sign the transaction for it to become valid and create the Platypus Blockchain $\Psi$. Since $f<m_0$, and only one transaction can be written in $\Omega$, we have that only with signatures from the correct processes it is enough to guarantee this condition. Therefore, the protocol terminates.
\end{proof}
\begin{theorem}
  \label{the:plcrexp}
  Let $\Psi$ be a Platypus Blockchain created by Algorithm~\ref{alg:plcr}, and let a correct process $p_i\in P_\Omega$. If $p_i\in P_\Psi$ then $p_i$ explicitly stated to be in $P_\Psi$ by sharing a signed Platypus Creation message $msg_i$.
\end{theorem}
\begin{proof}
  We prove this by contradiction. Suppose a $tx_{plcr}$ creation transaction such that some coins $Coins_i$ from process $p_i$ are included, without $p_i$ sending a signed Platypus Creation message $msg_i$. Suppose that transaction was written in $\Omega$, creating the Platypus Blockchain $\Psi$. For such transaction to be written in $\Omega$, it must be valid, i.e. it must hold at least $m_0$ signatures from validators. Since $f<m_0$, at least $m_0-f$ correct processes signed and verified such transaction (line~\ref{line:plcrval}). However, the correct processes could not validate such transaction without verifying its content (line~\ref{line:plcrver}), which includes verifying all the signed messages from all processes whose coins are involved in $tx_{plcr}$. Therefore, this is impossible without $p_i$ sending a signed Platypus Creation message $msg_i$.
  \end{proof}
\begin{corollary}
  \label{cor:01}
  Let $\Psi$ be a Platypus Blockchain created by Algorithm~\ref{alg:plcr}, and let a correct process $p_i\in P_\Omega$. If $p_i\in P_\Psi$ then $p_i$ explicitly stated to be in $P_\Psi$.
\end{corollary}
Notice that in Algorithm~\ref{alg:plcr} the 'only if' direction of Theorem~\ref{the:plcrexp} and Corollary~\ref{cor:01} is not necessarily true, should there be more than $m_v$ processes that reply to join. This does not affect the correctness of the protocol though.
\begin{lemma}
  \label{lem:mvv}
  Let $\Psi$ be a Platypus Blockchain created by Algorithm~\ref{alg:plcr}, then its Platypus Creation transaction $tx_{plcr}$ has $|V_\Psi|=m_v$ validators and was signed by $m_0$ of them.
\end{lemma}
\begin{proof}
  Given $f<m_0$ and $m_0$ are required for a Platypus Creation transaction to be valid, we have that some correct processes validated it. These correct processes verify that there are $m_v$ validators, and by Theorem~\ref{the:plcrexp} all validators explicitly stated they wanted to join as validators. Without enough signatures the algorithm does not terminate, since messages keep being sent (line~\ref{ref:plcrcont}), and $tx_{plcr}$ is not yet written in $\Omega$ (which is a condition in line~\ref{line:plcrcond}). By Theorem~\ref{the:plcrter} we know that the algorithm terminates. Thus, a valid $tx_{plcr}$ receives $m_0$ signatures, of which some processes could only have signed if $m_v$ processes were in the transaction as validators.
  \end{proof}

\begin{theorem}
  \label{the:ter}
  Algorithm~\ref{alg:plcl} guarantees the termination property.
\end{theorem}
\begin{proof}
  The protocol only waits for responses 4 times: to get coins from at least 
  $m_0$ signatures (line~\ref{line:sig}), and for the transaction to get in the Platypus Blockchain and parentchain (lines~\ref{line:cor},~\ref{line:par3} and~\ref{line:par2}). All these steps are independent of one another, i.e. not the same validators are required in each step. Therefore, we consider them independently. Since $m_0\geq 2m_v/3+1$, we have that, regardless of what the Adversary decides to do, $m_0$ correct nodes will eventually send enough signatures, and coins. Since we have both the Platypus Blockchain and parentchain consensus protocols are Byzantine Fault Tolerant, the calls that wait for a reply will terminate if $f<m_v/3$ and $f<n_v/3$, thus generating a response in the Platypus protocol (line~\ref{line:cor}), which could be either a COMMIT or an ABORT. Therefore, a correct process decides COMMIT or ABORT as the result of the call to $\lit{acsend(...)}$ in line~\ref{line:cor}. In either case, the protocol continues sending the proper transaction to the parentchain (lines~\ref{line:par3} and~\ref{line:par2}), which also terminates.
\end{proof}
\begin{lemma}
  \label{lem:01}
  In Algorithm~\ref{alg:plcl}, given a bulk close transaction listing a sequence $seq$ that process $p_i$ proposed to COMMIT, either all correct processes of the Platypus chain $\Uppsi$ decide ABORT to include the transaction in the Platypus Blockchain, or all correct processes decide COMMIT.
\end{lemma}
\begin{proof}
  We prove this by contradiction.  First, notice that, for a process to propose COMMIT on a Platypus Bulk Close transaction, it is necessary to provide a block where that transaction was written in the Platypus Blockchain. We consider the following network partition into three sets: $F$, the set of the adversary coalition of size $f<m_v/3$, $Q_1$ and $Q_2$. We consider that, at some point, all validators in $Q_1$ signed a block $b_1$ to validate a Platypus Bulk Close transaction, whereas validators in $Q_2$ validated a different block $b_2$ that spent from one of the same outputs (conflicting transactions).
  For one correct process to propose COMMIT, it is necessary that $b_1$ was validated by at least $m_0\geq 2m_v/3+1$ validators. Analogously, for one process to propose ABORT, it has to provide valid proof through a block $b_2$ validated by at least $m_0\geq 2m_v/3+1$ validators, in which some coins were spent from the owners claimed in the Platypus Bulk Close. A COMMIT proposal is undecided for as long as a valid ABORT is proposed, or enough validators validate the COMMIT attempt.

  In this case, we consider that one correct process proposes ABORT, meaning that it has a valid ABORT transaction, i.e. $b_2$ was validated by at least $2m_v/3+1$ validators. Therefore, $|Q_2\cup F|\geq 2m_v/3+1$. However, if another correct process committed to block $b_1$, then block $b_1$ has $2m_v/3+1$ validators. Thus, $|Q_1\cup F|\geq 2m_v/3+1$. Recall that $|F|=f<m_v/3$ and therefore $|Q_1|\geq m_v/3+1$ and $|Q_2|\geq m_v/3+1$, but this is impossible since $F\cup Q_1\cup Q_2=V_\Psi$ and $Q_1\cap Q_2=Q_1\cap F=F\cap Q_2=\emptyset$, and each account is only used once 
  . It follows that only $Q_2$ or only $Q_1$ had enough validators, and thus only some correct processes proposing and deciding ABORT (after which all will decide ABORT), or some processes deciding COMMIT (leading all other processes to decide COMMIT once they update their view of the childchain, since they do not decide ABORT) are possible.
\end{proof}
\begin{lemma}
  \label{lem:02}
  A Platypus Bulk Close transaction (COMMIT) can only be valid in $\Omega$ if it is already written in $\Psi$.
\end{lemma}
\begin{proof}
  For this, we assume that the transaction is sent to the parentchain without it being fully signed (i.e. beyond the threshold $m_0$) in the Platypus chain.
  A Byzantine process can try to send directly to the parentchain a not fully signed Platypus Bulk Close transaction (i.e. a Platypus Bulk Close transaction that was not written in the Platypus Blockchain). However, this transaction is not valid in the parentchain until it receives enough signatures. Notice that any process in the parentchain (i.e. Platypus Blockchain processes too) can eventually see this transaction, and generate a valid ABORT proof, or try to get it written in the Platypus Blockchain and then generate a valid COMMIT. Therefore, this proof is analogous to that of Lemma~\ref{lem:01}.
\end{proof}
\begin{theorem}
  \label{the:agr}
  Algorithm~\ref{alg:plcl} guarantees the agreement property.\\
\end{theorem}
\begin{proof}
  By Lemma~\ref{lem:02} we know that all COMMIT decisions are firstly written in $\Psi$. Then, Lemma~\ref{lem:01} shows that all processes in $P_\Psi$ reach the same decision to write in $\Psi$. We only have left the case that an ABORT is decided without it being written in the Platypus Blockchain $\Psi$. We need to prove that if that ABORT is decided then no process decided COMMIT. An ABORT outside of $\Psi$ can only happen if a process $p_i$ tried to COMMIT directly to $\Omega$ a Bulk close transaction that is not valid. Then, another process $p_j$ generated a valid Proof-of-Fraud included in an abort transaction, that ended up in an ABORT decision. Analogous to the proof of Lemma~\ref{lem:01}, we have a valid Proof-of-Fraud that gathers at least one conflicting transaction written in a previous block in $\Psi$, and therefore validated by at least $m_0$ validators. With the same approach used in Lemma~\ref{lem:01}, it is possible to prove that it is not possible for $p_j$ to propose a valid ABORT if one correct process $p_i$ decided COMMIT. Once a COMMIT is decided by enough processes, the funds go back to the Blockchain in the bulk close transaction of the sequence committed. Therefore, another sequence in another bulk close transaction will not be COMMIT-decided by any correct process. Hence, the agreement property is guaranteed.
\end{proof}
\begin{theorem}
  Algorithm~\ref{alg:plcl} guarantees the COMMIT-validity property.\\
\end{theorem}
\begin{proof}
  Lemma~\ref{lem:02} shows that the only way to get something committed is to first write it in $\Psi$, while Lemma~\ref{lem:01} proves that either all or no correct process decide COMMIT on a sequence. If no correct process proposes ABORT and, by Theorem~\ref{the:ter}, they guarantee termination, then they must COMMIT. 
\end{proof}
\begin{theorem}
  Algorithm~\ref{alg:plcl} guarantees the ABORT-validity property.\\
\end{theorem}
\begin{proof}
  If a correct process proposes ABORT in $\Psi$, then by Lemma~\ref{lem:01} all correct processes decide ABORT. All correct processes in $\Psi$ also agree on an ABORT generated to a COMMIT outside of $\Psi$, as already shown in the proof of the agreement property (Theorem~\ref{the:agr}).
\end{proof}

\begin{theorem}
  Algorithm~\ref{alg:plcl} guarantees the COMMIT-Privacy/Lightness property.  
\end{theorem}
\begin{proof}
  First, we consider the case that a Platypus Bulk Close transaction was successfully written in the parentchain (i.e. a COMMIT). W.l.o.g. we assume this to be the second transaction (after the Platypus Creation transaction) to be written in the parentchain relating this Platypus chain $\Uppsi$, i.e. that no previous Abort transactions were written. Let $t_c$ the time when the Platypus chain was created, $t_d$ the time when the Platypus chain was closed.
  This Platypus Bulk Close transaction has been validated in the Platypus Blockchain $\Psi$, verifying all the operations were correct. The parentchain processes that are not in the Platypus chain have no knowledge of the Platypus chain other than its Platypus Creation transaction that was written in $\Omega$.
  Therefore, a Platypus Bulk Close transaction with enough signatures from validators, and valid signatures, seems correct from the point of view of $\Omega$. Therefore, only this information, along with the list of coins and owners, is provided to the parentchain. This means that parentchain validators only stored the list of owners and coins at $t_c$, and received a different list of owners and their coins at $t_d$. They can tell which coins changed ownership between $t_c$ and $t_d$, but they cannot tell if there were more owners in between. Thus, they can only see the minimal transfers set.

  Whereas the COMMIT-Privacy/Lightness property considers COMMITs, if Abort transactions took place in between $t_c$ and $t_d$, a few more operations might be revealed to the parentchain to prove invalidity in Abort transactions. However, changing $t_c$ to the time of the last Abort, the proof remains valid.
\end{proof}

\begin{theorem}[Correctness]
  The Platypus Protocol solves the Offchain problem.
\end{theorem}
\begin{proof}[Proof]
  The proofs for Algorithm~\ref{alg:plcr} guarantee that $m_v$ validators are requested at all times (Lemma~\ref{lem:mvv}), all of which explicitly stated to participate as validators (Corollary~\ref{cor:01}), with guaranteed termination if there are enough validators $m_v$ (Theorem~\ref{the:plcrter}), i.e. the Platypus chain is properly bootstrapped and the security assumptions remain at the end of Algorithm~\ref{alg:plcr}. Once this bootstrapping takes place, the inner consensus of $\Psi$ guarantees the consensus properties given the assumption $f<m_v/3$ and the unforkability property, with the same set $m_v$ of validators and using the same $m_0$ as threshold for Byzantine behaviour (e.g. DBFT). Finally, given this bootstrapping and consensus protocol, we show above that Algorithm~\ref{alg:plcl}, which closes the Platypus chain, guarantees termination, agreement, ABORT-validity, COMMIT-validity and COMMIT-lightness/privacy. Therefore, Platypus solves the Offchain problem.
\end{proof}
The following theorem shows that our construction works in the strongest coalition the Adversary can form.

\begin{theorem}[Resilience optimality]
  \vspace{-0.55em}
  \label{the:imp}
  It is impossible to perform a transfer operation in an Offchain protocol with partial synchrony if $f\geq m_v/3$.
\end{theorem}
\begin{proof}
  We prove this by contradiction. If $f>n_v/3$ then the Adversary can corrupt the Blockchain, and thus the Offchain protocol is not correct. Thus, suppose $f\geq m_v/3$ while still $f<n_v/3$. Let there be at least one coin \textcent~is transferred from account $a$ to account $b$ in transaction $tx$ in the Offchain protocol $\Gamma$ (i.e. not the trivial case of closing after opening). We assume that there exists a correct Offchain protocol that solves the Offchain problem with such an Adversary. We look at the amount of validators the Blockchain protocol requires for the protocol to COMMIT, $m_0$. If the protocol had a threshold of $m_0> 2m_v/3$ signatures, it follows that some of the processes controlled by the Adversary should have agreed to such transaction. But the Adversary may decide not to validate, and thus the offchain protocol cannot continue, not performing any offchain transfer.

  Thus, $m_0$ has to be such that $m_0\leq 2m_v/3$. In such a case, consider a partition of validators $V_\Psi$ into $Q_1,\,Q_2$ of correct processes, and $F$ the set of processes controlled by the Adversary, such that $Q_1\cap Q_2=Q_1\cap F=F\cap Q_2=\emptyset$. Suppose there is another transaction $tx'$ that transfers the same coin \textcent~from account $a$ to $c,\; c\neq b$. Suppose $|Q_1|=|Q_2|=\frac{|V_\Psi|-|F|}{2}<m_v/3$. Since $Q_1\cup Q_2 \cup F = V_\Psi$, we have that $|Q_1|+|Q_2|+|F|=m_v$, meaning that $|Q_1|<m_v/3$ and $|Q_2|<m_v/3$. Therefore, $|F|+|Q_1|>2m_v/3$ and $|F|+|Q_2|>2m_v/3$. In this case, if $m_0<|F|+|Q_2|$, then $m_0<|F|+|Q_1|$ and thus it would be possible for the Adversary to validate $tx$ for $Q_1$ and $tx'$ for $Q_2$. Thus, $m_0$ must be such that $m_0>|F|+|Q_2|>2m_v/3$. This is a contradiction: we already showed above that $m_0$ should be such that $m_0\leq 2m_v/3$.
\end{proof}

\section{Theoretical Analysis}\label{sec:analysis}
In this section, we analyze the communication, message and time complexity of the Platypus protocol, ignoring the complexity of the underlying blockchain. 
We consider the calls to $\lit{acsend(\Psi,tx)}$ and $\lit{send(\{\Psi,\Omega\},tx)}$ to have the same complexities as one multicast to all validators $V_{\{\Psi,\Omega\}}$ of the blockchain that receives the transaction $tx$.
\subsubsection{Message complexity} The message complexity of Algorithms~\ref{alg:plcl},~\ref{alg:spin} and~\ref{alg:spou} is $\mathcal{O}(m_v^2)$ and that of Algorithm~\ref{alg:val} is $\mathcal{O}(m_p*m_v)$. We conjecture that the complexity could however be reduced to $\mathcal{O}(m_v)$ at some points, leveraging non-interactive aggregation of the validators signatures and messages, but certain calls to $\lit{acsend(...)}$ and $\lit{send(...)}$ would still have a complexity of $\mathcal{O}(m_v^2)$, as they can be executed by all processes. The same applies to Algorithm~\ref{alg:plcr}, with the exception that Phase 2 has a message complexity of $\mathcal{O}(m_p*n_p)$, thus being this one the complexity of Platypus.
\subsubsection{Communication complexity} The message size is $\mathcal{O}(m_p)$ 
 in lines~\ref{ref:plcrcont} and~\ref{line:comcop} of Algorithm~\ref{alg:plcr}, leading to a communication complexity of $O(\max\{m_p*n_p,\, m_v^3\})$, because of phases 2 and 3 of the algorithm. Line~\ref{line:vPoF} of Algorithm~\ref{alg:val} also has a message size of $\mathcal{O}(m_v)$, leading to a communication complexity of $O(m_p*m_v^2)$, although the set of validators can be removed if no punishments are considered. The rest of messages have constant size in all algorithms, thus their communication complexity is the same as their message complexity.
\subsubsection{Time complexity} The time complexity is $\mathcal{O}(m_v)$ due to phases 4 of Algorithm~\ref{alg:plcr}, and, Phase 2 of Algorithm~\ref{alg:plcl}. Algorithms~\ref{alg:val},~\ref{alg:spin} and~\ref{alg:spou} have constant time complexity. Again, we conjecture that, leveraging non-interactive aggregation, the time complexity can be reduced to constant time.

Note that these complexities are lower or comparable to consensus algorithms in the same model~\cite{crain2018dbft}.



\section{Improvements \& Discussion}
\label{sec:dis}
In this section, we consider additional features of the Platypus chain, and its usage for the general sidechains problem, which we also define.

\subsection{Crosschain payments}
A crosschain payment can be of two types, either a payment to a parentchain, or a payment through a parentchain to another childchain. With the above-shown protocol, a payment to a parentchain
would require a Platypus bulk close transaction, and a new Platypus creation transaction. We describe an extension of the protocol to perform payments without closing and reopening Platypus chains.
\subsubsection{Users' Splice-in \& Splice-outs}
Splice-in and Splice-out transactions allow users to get their funds into and out of the Platypus chain, respectively. 

  \textit{$\cdot$ Splice in.}
  Splicing in allows users to join a Platypus chain. Since this transaction takes place after the Platypus chain has been created, it requires some validation by both their sets of validators. Algorithm~\ref{alg:spin} shows the Splice in protocol for a process $p_i$ that wants to join $\Psi$. A splice in transaction $tx_{spin}$ must be written in both $\Omega$ and $\Psi$, after which the funds can only be spent in $\Psi$. 

  

  \begin{algorithm}[H]
    \small
    \caption{Splice in algorithm for process $p_i$}
    \label{alg:spin}
    \begin{algorithmic}[1]
      \Statex $\vartriangleright$ State of the algorithm
      \Statex $\Omega$, $\Psi$, $\Gamma$, the Blockchain, Platypus Blockchain and protocol
      \Statex $\mathds{C}_i$, coins that belong to process $p_i$
      \Statex $plid$, the Platypus chain identifier
      \Statex $tx_{spin}\gets\bot$, the splice in transaction
      \Statex \rule{0.4\textwidth}{0.4pt}
      \Statex $\vartriangleright$ $p_i$ creates and waits for transaction to write
        \State $tx_{spin} \; \gets \lit{createSpliceInTx(}\mathds{C}_i,plid\lit{)}$
        \State $tx_{spin} \; \gets \lit{sign_i(}tx_{spin}\lit{)}$
        \State $\Gamma\lit{.send(}\Omega, tx_{spin} \lit{)}$
        \State $\Gamma\lit{.send(} \Psi, tx_{spin} \lit{)}$
  \end{algorithmic}
\end{algorithm}
  \textit{$\cdot$ Splice out.}
  The same way users can splice into an existing Platypus chain, they can get their funds back in the parentchain. Again, this is a sensible operation that requires proper synchronization between both Platypus chain and parentchain so as to protect against fraud.

  
The splice out transaction allows processes to leave a Platypus chain before it is closed, retrieving their funds back in the parentchain. In this case, we require first the transaction to be finalized in $\Psi$ before being considered for the parentchain. 
  Algorithm~\ref{alg:spou} shows the splice out protocol for a process $p_i$. This protocol is rather a simplification of Algorithm~\ref{alg:plcl}. It creates and tries to write a splice out transaction $tx_{spou}$, that can be aborted with an abort transaction $tx_{abort}$.  

  \begin{algorithm}[H]
    \small
    \caption{splice out for process $p_i$}
    \label{alg:spou}
    \begin{algorithmic}[1]
      \Statex  $\vartriangleright$ State of the algorithm
      \Statex $\Omega$, $\Psi$, $\Gamma$, the Blockchain, Platypus Blockchain and protocol
      
      \Statex $\mathds{C}_i$, the coins that belong to process $p_i$
      \Statex $tx_{spou}\gets \bot$, the splice out transaction
      \Statex \rule{0.4\textwidth}{0.4pt}
      \Statex $\vartriangleright$ $p_i$ creates and waits for transaction to write in $\Psi$
      
      \State $tx_{spou} \; \gets \lit{createSpliceOutTx(}\mathds{C}_i\lit{)}$
      \State $tx_{spou} \; \gets \lit{sign_i(}tx_{spou}\lit{)}$

      \State $r\gets\Gamma\lit{.acsend(}\Psi,tx_{spou}\lit{)}$ \Comment{Get back $tx_{spou}$ or $tx_{abort}$} \label{lin:cor}
      \SmallIf{$r.type=\lit{ABORT}$}{ $\Gamma.\lit{send(}\Omega,r.tx_{abort}\lit{)}$}\label{line:par4}
      \EndSmallIf
      \SmallElseIf{$r.type=\lit{COMMIT}$}{$\Gamma.\lit{send(}\Omega,r.tx_{spou}\lit{)}$} \label{line:spoucommit}
\label{line:par5}
      \EndSmallElseIf
  \end{algorithmic}
\end{algorithm}
\subsubsection{Validators' Splice-in \& Splice-outs}

It is important to consider that the adversary should not gain enough relative power, either by splicing in or by correct validators splicing out. One way to guarantee this is by keeping the set of validators intact regardless of the funds each validator has after Platypus creation. This approach is similar to the Platypus bulk close transaction, and ensures correctness of the protocol, although it can be cumbersome for a validator to keep track and participate in a Platypus chain it no longer takes active part in. For this reason, an additional feature of the protocol might provide explicit delegation of the validator set to other users, similar to how consortium blockchains behave.


Another alternative may allow users and the set of validators to splice in and splice out in a permissionless, Proof-of-Stake based environment. In this set, validators should take great care at identifying the probability of an adversary gaining enough relative power, either through simple heuristics based on the funds at stake, or additional information, such as trust in other validators. If the probability of an adversary gaining enough relative power reaches a certain threat threshold, either by the validators set reducing significantly or any other information used for heuristics, validators can generate a Platypus bulk close transaction and safeguard all users' funds. This variation requires the assumption that the adversary never gains enough relative stake such that $stake(f)\geq stake(m_v)/3$.

\subsubsection{Crosschain payments with splice-in \& splice-outs}
\label{sec:cspio}
A crosschain payment in between two blockchains with Platypus is a payment of one user from/into an existing Platypus chain to/from its parentchain, or in between two Platypus chains that share a common parentchain. In section~\ref{sec:extcross}, we generalize such definition. Regardless of the particular conditions and assumptions for splice-ins and splice-outs, we illustrate in this section how these transactions would work.

$\cdot$ \textit{Crosschain payment from/to parentchain}. This case is trivial using Algorithm~\ref{alg:spin} or~\ref{alg:spou}, respectively. 

$\cdot$ \textit{Crosschain payment between childchains}. A crosschain payment between Platypus chains is performed with a splice out into the common parentchain, followed by a splice in into the recipient.

\subsection{Session Keys}
In the Platypus chain $\Uppsi$, we specify the requirement of one time accounts. This is to prevent a variant of the ABA problem, in which at time $t_1$ A transferred the coin to B, which in turn transferred it back to $A$ at time $t_2$. If $A$ tries to Platypus Bulk Close claiming ownership at time $t_3$, $B$ could ABORT with a valid proof of a spent from time $t_2$. While using one time accounts already solves this problem, since $A$ would actually use $A_1$ for $t_1$ and $A_2$ for $t_2$ and $t_3$, the approach can also be solved while allowing $A_1=A_2$ by introducing further data in Platypus Bulk Close transactions, such as the merkle tree of the state of the Platypus Blockchain. 

\subsection{Platypus for Sidechains}
\label{sec:extcross}
The childchain definition done in section~\ref{sec:mod} can easily be generalized for sidechains by clearly decoupling $\Psi$ from the protocol, and stating different sets for them $P\neq Q$ instead of $P \varsubsetneq Q$. We define sidechain protocols as a superset of offchain protocols, defined in Section~\ref{sec:mod}. A sidechain protocol allows to perform a payment across blockchains, i.e. a \textit{crosschain payment}. If two or more blockchains intend to perform crosschain payments, we refer to them as being sidechains. They may or may not be in a parent-child hierarchy.

$\cdot$ \textit{Sidechain protocol.}  
Given two blockchains, $\Omega$ of $P$ processes and $\Psi$ of $Q$, $P\neq Q$ a sidechain protocol $\Pi$ is an offchain protocol that enables transfers in between all accounts $p_a,q_a$ such that {\small $\rho(q_a)=q \in Q, \rho(p_a)=p\in P$}. To reflect $\Omega$ and $\Psi$ being independent, and this possibility of transferring, we define the following property:

$\>-$ COMMIT-Matching Knowledge: If a correct process decides COMMIT on a sequence $seq$ of transfer operations in $\Pi$ between $\Omega$ and $\Psi$, then $\forall p\in P,\; p$ knows a subset $seq_1$ and $\forall q\in Q,\; q$  knows a subset $seq_2$, such that $seq_1$ and $seq_2$ are two minimal transfer sets, $seq_2\cap seq_1=\emptyset$, and it exists one surjective application $f:seq_1\times seq_2\rightarrow TR^-(seq_1\cup seq_2)\cup\{0\}$ defined as follows:
    \begin{equation}
      \hspace{-4.7em}
      \small
      f(a\,TR\,b,c\,TR\,d)=  \left \{\begin{aligned}(a\,TR\,d) &\;\text{if}\; \rho(b)=\rho(c)=p_i\\ 0\; &\,\text{otherwise} \end{aligned}\right \}
    \end{equation}
 Also, since the coins are different in different blockchains, we identify coins by their value when calculating the minimal transfer set $TR^-$. Intuitively, for a transaction in $seq_1$ exists a transaction in $seq_2$ such that both are transitive (that is, the receiver of one is the sender of the other). If that was not to happen, then some of the transactions in $seq_1$, or in $seq_2$, would have nothing to do with a payment in between two sidechains.

  If $Q\varsubsetneq P$ then $seq_1$ is just a set of idempotent transfers of the form $a\, TR\, b$, with $\rho(a)=\rho(b)$, since all $p\in Q$ are also in $P$, and thus COMMIT-Privacy/Lightness is a particular case scenario of the COMMIT-Matching Knowledge property. 

  Similarly, if $P\not\supseteq Q,\,P\not\subseteq Q$ then $\Omega$ and $\Psi$ are not in the parent-child chain hierarchy. 

  $\cdot$ \textit{Crosschain payments.} This is solved by our protocol if both sidechains have a common parentchain, as shown in section~\ref{sec:cspio}. In general, for a crosschain payment between two unrelated Blockchains $\Omega_1$ and $\Omega_2$, with sets of validators $V_{\Omega_1}$ and $V_{\Omega_2}$, they can perform the payment manufacturing an additional Blockchain $\Phi$:

  $\>-$ Create a common parentchain $\Phi$ with $V_\Phi \supseteq V_{\Omega_1}\cup V_{\Omega_2}$, extend both their Blockchains to adopt the Platypus protocol, and perform the payment as explained in section~\ref{sec:cspio}. In this case, if the adversary tries to double spend the crosschain payment in $\Omega_2$, or in $\Omega_1$, then, as long as $f<|V_\Phi|/3$, the funds will remain in the parentchain $\Phi$.

  $\>-$ Create a common Platypus chain $\Phi$, with $V_\Phi\subseteq V_{\Omega_1}\cap V_{\Omega_1}$, and perform the payment. In such a case, should $f<|V_{\Omega_1}|$ and $f<|V_{\Omega_2}|$, then the adversary could not double spend the funds in $\Phi$ and splice out to both $\Omega_1$ and $\Omega_2$.
  \subsection{Attacks}
  Many of the common attacks for synchronous offchain protocols are not applicable in the partially synchronous Platypus, as they exploit timelocks, such as forced expiration spam~\cite{poon2016bitcoin}, balance disclosure attacks~\cite{2019difficulty} or stale attacks~\cite{2019scalable}, among others. The colluding validators attack is possible, should the adversary be such that $f\geq m_v/3$, as shown by Theorem~\ref{the:imp}.

  However, it is still possible to perform a colluding validators attack in this protocol, if the Adversary gains enough power to influence the parentchain or the childchain. Again, this can only happen if $f\geq m_v/3$ or $f\geq n_v/3$, and since the offchain protocol we propose composes a partially synchronous Blockchain, the impossibility result shown in theorem~\ref{the:imp} already shows this is the best possible case for a partially synchronous offchain protocol.

  We also introduce the \textit{ABA-transfer} attack. If a coin \textcent~was transferred from process $p$ to process $q$, and later on again to process $p$, $q$ can try to ABORT any close/splice out in which \textcent~does not belong to him, by using as proof the deprecated transfer $p \,TR\,p$. To cope with this attack, we use session keys in this document, as mentioned in Section~\ref{sec:mod}, thus having two different accounts. Another possible solution involves committing to merkle trees and requiring any ABORT to provide a merkle tree $T$ such that the merkle tree $T'$ of the COMMIT attempt is included $T'\subseteq T$ as part of the Proof-of-Fraud.

  \subsection{Accountability: Punishment through Abort}
It is easy to see that, if more than a third of the validators are selfish, instead of correct, they can collude to an attack. Validators may collude to create Platypus Creation transactions with other users, and then Platypus Bulk Close, claiming all the funds (effectively stealing those funds). This is not at conflict with this current model, in which there are correct or Byzantine processes, such that $f<m_v/3$.

Other selfish users may try to fork the Platypus Blockchain $\Psi$, and perform a double spend. Accountability plays a major role in attacks of this type, as shown by~\cite{cryptoeprint:2019:587}. We are currently developing and extension of this protocol in the rational model.

Intuitively, since the abort function gathers Proofs-of-Fraud (PoFs) for each validator that was a fraudster, the protocol can create disincentives for Byzantine behaviour through punishments when Abort transactions with PoFs are written in the parentchain. This disincentivizes selfish actors from colluding to an invalid state, whereas it provides the same guarantees and correctness in the Byzantine Fault Tolerant model. The rational model requires however further modifications and assumptions. 
\subsubsection{Creation and Destruction with Accountability}
Apart from disincentives, the Platypus protocol can modify creation and closing to make it harder, or even impossible (in the case of creation), for selfish users to collude to attack the protocol in these steps. For example, by requiring every single signature at creation and destruction, other users can neither lock coins nor choose which coins they claim when destroying the childchain.
The reader can spot that this requirement might come at the cost of non-termination: one single validator or user not signing can lock the protocol. To cope with this, creation and destruction can be divided into multiple transactions, each of which require at least $m_0$ validators, and all signatures of everybody who is joining/exiting the Platypus chain. Again, this is a work in progress for the rational model, and further modifications are required to guarantee correctness in this environment.


\section{Related Work}
\label{sec:rel}


\subsubsection{Sidechains \& Childchains}
Childchains were first introduced with the concept of sidechains~\cite{back2014enabling}. A sidechain has a broader definition than a childchain 
has been used to execute crosschain payments 
without a parent-child hierarchical structure.
%
Childchains were first formalized in~\cite{gazi2019proof} where the  
authors propose 
an efficient childchain protocol in a semi-synchronous model.
Unfortunately, their notion of semi-synchronous communication considers that every messages get delivered in a non-null bounded amount of time $\Delta$, which remains a synchrony assumption~\cite{dwork1988consensus}. The term `semi' is used by the authors to denote the fact that the bound $\Delta$ is not null. Note that this notion differs from partial synchrony~\cite{dwork1988consensus} where the bound is unknown.

\subsubsection{Crosschain payments}
Many protocols propose generic crosschain payments. 
Atomic crosschain swaps~\cite{nolan,herlihy2018atomic,zamyatinxclaim,zakhary2019atomic} typically rely on Hashed Timelock Contracts~\cite{HTLCs} that are synchronous, while others focus on a crash failure model, rather than a Byzantine one~\cite{zakhary2019atomic}. 
Chain relays bridge information from one Blockchain by writing it in a different Blockchain, such as an Ethereum smart contract storing Bitcoin headers~\cite{chow2016btc}. 
Consensus-based crosschain interactions~\cite{wood2016polkadot,kwon2014tendermint} are the closest to our proposal, with some of them falling in the sidechain category. Polkadot~\cite{wood2016polkadot} reuses the idea to manufacture a common parentchain to Blockchains, in order to perform payments asynchronously, although 
it was not proved correct. 

Crosschain deals~\cite{herlihy2019cross} allow for auctions or relaying payments.
The authors outline both a synchronous and a partially synchronous protocols.
Unlike our problem, the crosschain deals problem tolerates that the protocol aborts even if the only processes proposing abort are Byzantine. Our problem disallows such an execution as it could prevent a correct process from cashing out.
Our implementation ensures this execution cannot happen by requiring every abort transaction to contain a valid proof of fraud. 
This makes our offchain problem the first Byzantine fault tolerant variant of the atomic commit problem~\cite{cachin2011introduction} that has only been defined to our knowledge in a crash model.

\subsubsection{Offchain protocols}
State and payment channels~\cite{poon2016bitcoin,decker2015fast} were the first offchain proposals for Blockchains. The Lightning Network~\cite{poon2016bitcoin}, a network of channels that relay payments offchain, is the most notable of the offchain proposals, while other similar offchain payment networks have been proposed\cite{khalilnocust,miller2017sprites}, some of which work in asynchronous communications~\cite{avarikioti2019brick,lind2018teechain}. While channel networks scale, they are still limited to the amount of transactions allowed to open and close each of its channels. Lightning Factories allow users to open several channels at once while preserving constant lock-in time~\cite{2019scalable}. Other works also target more scalability than payment channels through factory-like constructions under different systems and assumptions~\cite{burchert2018scalable}. All factories and channels require of all involved users to explicitly sign to perform transfers, impacting performance. PLASMA is the most known childchain construction, proposed for Ethereum~\cite{poon2017plasma}. 
It provides the first childchain protocol with fraud detection. 


All these offchain protocols are synchronous, which could make them vulnerable to the
Balance Disclosure attack~\cite{2019difficulty} 
that
discloses the balances of other users in the Lightning Network, while the Stale Channel/Factory attack locks balances of all users in a channel/factory~\cite{2019scalable}. 
There is therefore great interest in achieving offchain scalability and privacy in partial synchrony.



\section{Conclusion}
\label{sec:conc}
The Platypus chain is the first childchain that does not assume synchrony or a trusted execution environment.
We prove its correctness, and discuss its extensions and applications for scalability and for secure crosschain payments. 
Finally, we showed that our protocol is correct and resilience optimal.
%
As future work, we would like to cope with more than $n/3$ rational processes. 

{\small
\bibliographystyle{abbrv}
\bibliography{main}
}



\end{document}